\documentclass[a4paper]{easychair}

\usepackage{graphicx}
\usepackage{semantic}
\usepackage{listings,multicol}
\usepackage[utf8x]{inputenc}
\usepackage{lstcoq}
\usepackage{subcaption}
\usepackage{enumitem}
\usepackage{stmaryrd}
\usepackage{amssymb}

\usepackage{xcolor}

\definecolor{ltblue}{rgb}{0,0.4,0.4}
\definecolor{dkblue}{rgb}{0,0.1,0.6}
\definecolor{dkgreen}{rgb}{0,0.35,0}
\definecolor{dkviolet}{rgb}{0.3,0,0.5}
\definecolor{dkred}{rgb}{0.5,0,0}
\definecolor{bggray}{gray}{0.95}

\newcommand{\icode}[1]{\lstinline[mathescape,basicstyle=\footnotesize]!#1!}

\newtheorem{theorem}{Theorem}

\newtheorem{lemma}{Lemma}
\newtheorem{corollary}{Corollary}
\newtheorem{definition}{Definition}

\newcommand\oak[1]{{\color{dkblue}{#1}}}

\newcommand\mcoq[1]{{\textcolor{dkgreen}{#1}}}

\reservestyle{\command}{\mathtt}
\command{if[if],then[then],else[else],let[let],in[in],case[case],return[return],of[of],fix[fix],
fun[fun],Set[Set],forall[forall],match[match],as[as],in[in],with[with],end[end],
NotEnoughFuel[NotEnoughFuel],Ok[Ok],EvalError[EvalError]}

\newcommand\lift[1]{\uparrow_1#1}
\newcommand\liftn[2]{\uparrow_#1#2}
\newcommand\branchname{\mathtt{branch}}
\newcommand\branch[3]{\branchname(#1,#2,#3)}
\newcommand\find[2]{\mathtt{find}(#1,#2)}
\newcommand\ctorargs[1]{\mathtt{ctor\_args}(#1)}
\newcommand\ctorname[1]{\mathtt{ctor\_name}(#1)}
\newcommand\resolveindname{\mathtt{resolve\_ind}}
\newcommand\resolveind[2]{\resolveindname(#1,#2)}
\newcommand\resolvectorname{\mathtt{resolve\_ctor}}
\newcommand\resolvector[3]{\resolvectorname(#1,#2,#3)}
\newcommand\fromvalname{\ensuremath\mathtt{of\_val}}
\newcommand\fromval[1]{\fromvalname(#1)}

\newcommand\expT[2]{\llbracket#1\rrbracket^t_{#2}}
\newcommand\tyT[1]{\llbracket#1\rrbracket^T}

\newcommand\defn{\overset{def}{=}}

\newcommand\type[1]{\ensuremath\mathtt{#1}}

\newcommand\evaloak[3]{\ensuremath\mathtt{eval}_{#1}^{#2}(#3)}
\newcommand\evaltypename{\ensuremath\mathtt{eval\_ty}}
\newcommand\evaltype[1]{\evaltypename(#1)}
\newcommand\matchpatname{\ensuremath\mathtt{match\_pat}}
\newcommand\matchpat[5]{\matchpatname(#1,#2,#3,#4,#5)}
\newcommand\validatename{\ensuremath\mathtt{validate}}
\newcommand\validate[3]{\validatename(#1,#2,#3)}
\newcommand\validatebranchesname{\ensuremath\mathtt{validate\_branches}}
\newcommand\validatebranches[2]{\validatebranchesname(#1,#2)}

\newcommand\vClosLam[4]{\ensuremath\mathtt{vClosLam}(#1,#2,#3,#4)}
\newcommand\vClosFix[6]{\ensuremath\mathtt{vClosFix}(#1,#2,#3,#4,#5,#6)}
\newcommand\vConstr[3]{\ensuremath\mathtt{vConstr}(#1,#2,#3)}
\newcommand\vTy[1]{\ensuremath\mathtt{vTy}(#1)}
\newcommand\vTyClos[3]{\ensuremath\mathtt{vTyClos}(#1,#2,#3)}

\newcommand\id[1]{\ensuremath\mathit{#1}}
\newcommand\rev[1]{\ensuremath\mathtt{rev}({#1})}

\newcommand\map[2]{\ensuremath\mathtt{map}~#1~#2}

\newcommand\acorn{\textsc{Acorn}}
\newcommand\oaklang{\textsc{Oak}}
\newcommand\acornkernel{\ensuremath{\lambda_\mathrm{smart}}}

\newcommand\mcoqsubst[2]{#1\big\{#2\big\}}

\newcommand\vect[1]{\overrightarrow{#1}}

\usepackage{balance}
\usepackage[toc,page]{appendix}

\title{ConCert: A Smart Contract Certification Framework in Coq}
\titlerunning{ConCert: A Smart Contract Certification Framework in Coq}

\author{Danil Annenkov \and Jakob Botsch Nielsen \and Bas Spitters}
\institute{Concordium Blockchain Research Center, Aarhus University}
\authorrunning{Annenkov, Nielsen, Spitters}

\begin{document}
\maketitle

\begin{abstract}
  We present a new way of embedding functional languages into
  the Coq proof assistant by using meta-programming. This allows
  us to develop the meta-theory of the language using the
  deep embedding and provides a convenient way for reasoning about concrete
  programs using the shallow embedding.
  We connect the deep and the shallow embeddings by a
  soundness theorem.
  As an instance of our approach, we develop an embedding of a core
  smart contract language into Coq and verify several important
  properties of a crowdfunding contract based on a previous
  formalisation of smart contract execution in blockchains.
\end{abstract}

\section{Introduction}\label{sec:intro}
The concept of blockchain-based smart contracts has evolved in several ways
since its appearance. Starting from the restricted and non-Turing-complete
Bitcoin script\footnote{Bitcoin: A peer-to-peer electronic cash
  system.\\ \url{https://bitcoin.org/bitcoin.pdf} } designed to validate
transactions, the idea of smart contracts expanded to fully-featured languages
such as Solidity running on the Ethereum Virtual Machine (EVM).\footnote{Ethereum's
  white paper:\\ \url{https://github.com/ethereum/wiki/wiki/White-Paper}} Recent
research on smart contract verification discovered the presence of multiple
vulnerabilities in many smart contracts written in
Solidity~\cite{Luu:2016,Sergey:ConcurrentPerspective}. Several times the issues
in smart contract implementations resulted in huge financial losses (for
example, the DAO contract and the Parity multi-sig wallet on Ethereum). The
setup for smart contracts is unique: once deployed, they cannot be changed and
any small mistake in the contract logic may lead to serious financial
consequences. This shows not only the importance of formal verification of smart
contracts but also the importance of principled programming language
design. The third generation smart contract languages tend to employ the functional
programming paradigm. A number of blockchain implementations have already
adopted certain variations of functional languages as a smart
contract language. These languages range from minimalistic and low-level
(Simplicity~\cite{O'Connor:Simplicity},
Michelson\footnote{\url{https://www.michelson-lang.com/}}),
intermediate (Scilla~\cite{Sergey:Scilla}) to fully-featured
OCaml- and Haskell-like languages (Liquidity~\cite{Liquidity},
Plutus~\cite{Chapman:PlutusCore,Jones:UnravelingRecursion}).
There is a very good reason for this tendency. Statically typed functional
programming languages can rule out many errors. Moreover, due to the absence
(or more precise control) of side effects programs in functional languages
behave like mathematical functions, which facilitates reasoning about
them. However, one cannot hope to perform only stateless computations: the state
is inherent for blockchains. One way to approach this is to limit the ways of
changing the state. While Solidity allows arbitrary state modifications at any
point of execution, many modern smart contract languages represent smart
contract execution as a function from a current state to a new state. This
functional nature of modern smart contract languages makes them well-suited for
formal reasoning.

The Ethereum Virtual Machine and the Solidity smart contract language remain one
of the most used platforms for writing smart contacts. Due to the permissiveness
of the underlying execution model and the complexity of the language,
verification in this setting is quite challenging. On the other hand, many
modern languages such as \acorn{},\footnote{The \acorn{} language is an ML-style
  functional smart contract language currently under development at the
  Concordium Foundation.} Liquidity and Scilla, offer a different execution
model and a type system allowing to rule out many errors through type
checking. Of course, many important properties are not possible to capture even
with powerful type systems of functional smart contract languages. For that
reason, to provide even higher guarantees, such as functional correctness, one
has to resort to stronger type systems/\-logics for reasoning about programs and
employ deductive verification techniques. Among various tools for that purpose
proof assistants provide a versatile solution for that problem.

Proof assistants or interactive theorem provers are tools that allow users to
state and prove theorems interactively. Proof assistants often
offer some degree of proof automation by implementing decision and semi-decision
procedures, or interacting with automated theorem provers (SAT and SMT
solvers). Some proof assistants allow for writing user-defined automation
scripts, or write extensions using a plug-in system. This is especially
important, since many properties of programs are undecidable and providing
users with a convenient way of interactive proving while retaining a possibility
to do automatic reasoning makes proof assistants very flexible tools for
verification of smart contracts.

Existing formalisations of functional smart contract languages mostly
focus on meta-theory (Plutus~\cite{Chapman:PlutusCore},
Simplicity~\cite{O'Connor:Simplicity}) or meta-theory and verification
using the deep embedding (Michelson~\cite{Mi-Cho-Coq}). An exception
is Scilla~\cite{Sergey:Scilla}, which features verification of
particular smart contracts in the Coq proof assistant by means of
shallow embedding by hand. Simplicity~\cite{O'Connor:Simplicity} is a
low-level combinator-based functional language and its formalisation
allows for translating from deep to shallow embeddings for purposes of
meta-theoretic reasoning. None of these developments combines deep and
shallow embeddings for a \emph{high-level} functional smart contract
language in one framework or provide an automatic way of converting
smart contracts to Coq programs for convenient verification of
concrete smart contracts. We are making a step towards this direction
by allowing for deep and shallow embeddings to coexist and interact in
Coq.

The contributions of this paper are the following:
\begin{enumerate}
\item We develop an approach to verify properties of functional
  programming languages and of individual programs in one framework. In
  particular, this approach works for functional smart contract
  languages and concrete contracts.
\item We describe a novel way of combining deep and shallow embeddings using
  the meta-programming facilities of Coq (MetaCoq~\cite{Anand:TemplateCoq}).
\item As an instance of our approach, we define the syntax and semantics of
  \acornkernel{} --- a core subset of the \acorn{} language (the deep embedding)
  and the corresponding translation of \acornkernel{} programs into Coq
  functions (the shallow embedding).
\item We prove properties of a crowdfunding contract given as a deep
  embedding (abstract syntax tree) of a \acornkernel{} program.
\item We integrate our shallow embedding with the smart contract execution
  framework~\cite{Interactions} allowing for proving safety and temporal properties
  of interacting smart contracts.
\end{enumerate}
We discuss the details of our approach in Section~\ref{sec:approach},
provide an example of a crowdfunding contract verification in
Section~\ref{sec:crowdfunding}. In Section~\ref{sec:oak-stdlib} we
apply our framework to verify a \texttt{List} module of the \acorn{}
standard library and discuss how our development integrates with the
execution framework in Section~\ref{sec:exec-framework}. Theorems from
Section~\ref{sec:soundness} and lemmas from
Sections~\ref{sec:crowdfunding},~\ref{sec:oak-stdlib} and~\ref{sec:exec-framework} are
proved in our Coq development and available at \url{https://github.com/AU-COBRA/ConCert/tree/artefact}.

\begin{figure*}
  \includegraphics[height=3cm]{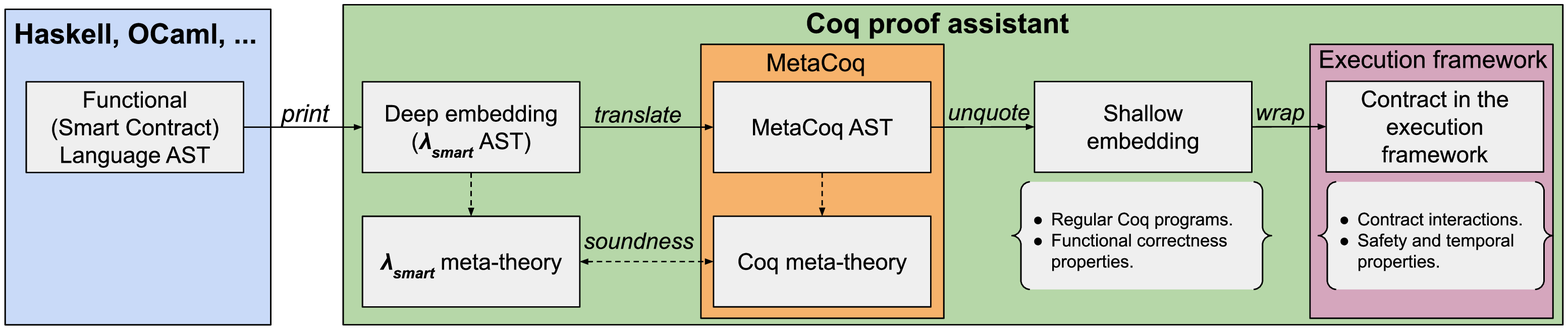}
  \caption{The structure of the framework}\label{fig:structure}
\end{figure*}

\section{Our Approach}\label{sec:approach}
There are various ways of reasoning about properties of a functional programming
language in a proof assistant. First, let us split the properties into two groups:
meta-theoretical properties (properties of a language itself) and properties of
programs written in the language. Since we are focused on functional smart
contract languages and many proof assistants come with a built-in functional
language, it is reasonable to assume that we can reuse the programming language
of a proof assistant to express smart contracts and reason about their
properties. A somewhat similar approach is taken by the authors of the hs-to-coq
library~\cite{Spector-Zabusky:TotalHaskell}, which translates total Haskell
programs to Coq by means of source-to-source transformation. Unfortunately, in
this case, it is impossible to reason about the correctness of the translation.

We would like to have two representations of functional programs within the same
framework: a deep embedding in the form of an abstract syntax tree (AST), and a
shallow embedding as a Coq function. While the deep embedding is suitable for
meta-theoretical reasoning, the shallow embedding is convenient for proving
properties of concrete programs. We use the meta-programming facilities of the
MetaCoq plug-in~\cite{Anand:TemplateCoq} to connect the two ways of reasoning
about functional programs.

The overview of the structure of the framework is given in
Figure~\ref{fig:structure}. As opposed to source-to-source translations in the
style of hs-to-coq~\cite{Spector-Zabusky:TotalHaskell} and
coq-of-ocaml\footnote{The coq-of-ocaml project page:
  \url{https://github.com/clarus/coq-of-ocaml}} we would like for all the
non-trivial transformations to happen in Coq. This makes it possible to reason
within Coq about the translation and formalise the required meta-theory for the
language. That is, we start with an AST of a program in a functional
language implemented in Haskell, OCaml or some other language, then we generate
an AST represented using the constructors of the corresponding inductive type in
Coq (deep embedding) by printing the desugared AST of the program. By printing
we mean a recursive procedure that converts the AST into a string consisting of
the constructors of our Coq representation. The main idea is that this procedure
should be as simple as possible and does not involve any non-trivial
manipulations since it will be part of a trusted code base. If non-trivial
transformations are required, they should happen within the Coq
implementation.



\subsection{MetaCoq}
The MetaCoq project~\cite{Anand:TemplateCoq} brings together several subprojects
united by the use of meta-programming and formalisation of Coq's meta-theory in Coq.
In particular, relevant for this project:
\begin{itemize}
  \item Template Coq --- adds meta-programming facilities to Coq. That is, it
    provides a way to \emph{quote} Coq definitions by producing an AST
    represented as an inductive data type \icode{term} in Coq, and
    \emph{unquote} a well-formed inhabitant of \emph{term} back to a Coq
    definition.
  \item PCUIC --- formalisation of the meta-theory of Polymorphic
    Cumulative Calculus of Inductive Constructions (PCUIC), an
    underlying calculus of Coq.\footnote{From now on, we will use
      MetaCoq to refer both to the quote/unquote functionality and to
      the formalisation of meta-theory.}
\end{itemize}

These features of MetaCoq have been used for defining various syntactic
translations from Calculus of Inductive Constructions (CIC) to itself
(e.g. parametricity translation~\cite{Anand:Parametricity}), developing a
certified compiler CertiCoq~\cite{Anand:CertiCoq} and for certifying extraction
of Coq terms to the untyped lambda-calculus~\cite{ForsterKunze:extraction}

Let us consider a simple example demonstrating the quote/unquote functionality.
\begin{lstlisting}
(* Quote *)
Quote Definition id_nat_syn := (fun x : nat => x).
Print id_nat_syn.
(* tLambda (nNamed "x")
    (tInd (mkInd "nat" 0) []) (Ast.tRel 0) : term *)

(* Unquote *)
Make Definition plus_one :=
  (tLambda (nNamed "x") (tInd (mkInd "nat"  0 ) [])
           (tApp (tConstruct (mkInd "nat" 0) 1 []) [tRel 0])).
Print plus_one.
(* fun x : nat => S x : nat -> nat *)

\end{lstlisting}

Our use of MetaCoq explores a new way of using meta-programming in
Coq. All existing use cases follow (roughly) the following procedure:
start with a Coq term, quote it and perform certain transformations
(e.g. syntactic translation, erasure, etc.). In our approach, we go in
the different direction: starting with the AST of a program in a
functional language we want to reason about, through a series of
transformations we produce a MetaCoq AST, which is then unquoted into
a program in Coq's Gallina language (shallow embedding). The
transformations include conversion from the named to the nameless
representation (if required) and translation into the MetaCoq AST. The
deep embedding also serves as input for developing meta-theory of the
functional language.

\subsection{The \acornkernel{} Language}
As an instance of our approach, we develop an embedding of the
``core'' of the \acorn{} smart contract language into Coq. We call
this core language \acornkernel{}. This language contains all the
essential features of a realistic functional language: System F type
system, inductive types, general recursion and pattern-matching. The
grammar of the language is given below.
\begin{eqnarray*}
  \oak{\tau}~\oak{\sigma} & ::= & {\oak{\hat{i}}} ~|~ \oak{I} ~|~ {\oak{\forall A.\tau}} ~|~ \oak{\tau~\sigma} ~|~{\oak{\tau \rightarrow\sigma}} \\
  \oak{p} & ::= & \oak{C~x_1~\ldots~x_n}\\
  \oak{e} & ::= & \oak{\overline{i}} ~|~{\oak{\lambda x : \tau.e}} ~|~ {\oak{\Lambda A.e}} ~|~
  {\oak{\<let>~x:~\tau~=~e_1~\<in>~e_2}} ~|~ \oak{e_1~e_2} \\
  && |~\oak{\<case>~e~:I~\tau_1 \ldots \tau_n~\oak{\<return>}~\sigma~\<of>}\\
    &&\;\;\oak{p_1\rightarrow e_1 ; \ldots ; p_m \rightarrow e_m}\\
  && |~ \oak{C_I}~|~ \oak{\<fix>~f~x : \tau_1 \rightarrow \tau_2 ~=~ e} ~|~ \oak{\tau}
\end{eqnarray*}
Here $\oak{I}$, $\oak{C}$, $\oak{A}$, $\oak{x}$ and $\oak{f}$ range over
strings representing names of inductive types, constructors, type variable
names, variable names and fixpoint names respectively. We use de Bruijn indices
to represent variables both in expressions and types (denoted
$\oak{\overline{i}}$ and $\oak{\hat{i}}$ respectively). Textual names $\oak{A}$, $\oak{x}$ and
$\oak{f}$ are only needed for decoration purposes making resulting Coq code
more readable. Note, that \acornkernel{} expressions are extensively annotated with typing
information. For instance, we annotate lambda abstractions with the domain type,
for fixpoints we store the types for the domain and for the codomain. Moreover,
$\<case>$-expressions require explicit type of branches.


\begin{figure*}
\begin{subfigure}[t]{.5\textwidth}
\begin{lstlisting}
  $\evaloak{\Sigma,\rho}{0}{e}$ $\defn$ $\<NotEnoughFuel>$
  $\evaloak{\Sigma,\rho}{S~n}{e}$ $\defn$ $\<match>~e~\<with>$
                   $\ldots$
                   | $\oak{\<fix>~f~x~:~\tau_1 \rightarrow \tau_2 ~=~ e} =>$
                       $\oak{\tau_1'} <- \evaltype{\oak{\tau_1}}$;
                       $\oak{\tau_2'} <- \evaltype{\oak{\tau_2}}$;
                       $\validate{\rho}{2}{\oak e}$;
                       Ok ($\vClosFix{\rho}{f}{x}{\oak{\tau_1'}}{\oak{\tau_2'}}{\oak e}$)
                   | $\oak{e_1~e_2} => v_2 <- \evaloak{\Sigma,\rho}{n}{\oak{e_2}}$;
                              $v_1 <- \evaloak{\Sigma,\rho}{n}{\oak{e_1}}$;
                              $\<match>~v_1~\<with>$
                              $\ldots$
                              | $\vClosFix{\rho}{f}{x}{\oak{\tau_1}}{\oak{\tau_2}}{\oak{e}} => $
                                 $\<if>~\mathtt{isConstr}(v_2)~\<then>~\evaloak{\Sigma,v_2::v_1::\rho}{n}{\oak e}$
                                                   $\<else>~\<EvalError>$
                              $\<end>$
\end{lstlisting}
\end{subfigure}
\begin{subfigure}[t]{.45\textwidth}
  \begin{lstlisting}
       | $\oak{\<case>~e~:I~\tau_1 \ldots \tau_k~\oak{\<return>}~\sigma~\<of>~\oak{\id{bs}}} =>$
                      $\validatebranches{\rho}{\oak{\id{bs}}}$;
                      $\_ <- \evaltype{\oak{\sigma}}$;
                      $\_ <- \mathtt{monad\_map}~\evaltypename~[\oak{\tau_1}; \ldots ;\oak{\tau_k}]$;
                      $v <- \evaloak{\Sigma,\rho}{n}{\oak{e}}$;
                      $\<match>~v~\<with>$
                      | $\vConstr{\oak{I'}}{\oak{C}}{\oak{\id{args}}} =>$
                          $\<let> (\_,tys) := \resolvector{\Sigma}{\oak{I'}}{\oak{C}}~\<in>$
                          $\<if>~(\oak{I}=\oak{I'})~\<then>~$
                              $\oak{e'} <- \matchpat{C}{n}{\oak{\id{tys}}}{\oak{\id{args}}}{\oak{\id{bs}}}$;
                              $\evaloak{\Sigma,\rev{\oak{\id{args}}} +\!\!+ \rho}{n}{\oak{e'}}$
                           $\<else>~\<EvalError>$
                      | _ => $\<EvalError>$
                       $\<end>$
     $\<end>$
  \end{lstlisting}
\end{subfigure}
  \caption{\acornkernel{} interpreter.}\label{fig:interpreter}
\end{figure*}

The semantics of \acornkernel{} is given as a definitional
interpreter~\cite{Reynolds1998}. This gives us an executable semantics
for the language. Moreover, since the core language we consider is
sufficiently close to \acorn{}, our interpreter can serve as a
\emph{reference} interpreter. The interpreter is implemented in an
environment-passing style and works both with named and nameless
representations of variables. Due to the potential non-termination, we
define our interpreter using the \emph{fuel idiom}: by structural
recursion on an additional argument (a natural number). The
interpreter function has the following type:\footnote{In our
  development, the interpreter supports two modes of evaluation: with
  named variables and with de Bruijn representation of variables
  (nameless). For the purposes of this paper, we will focus on the
  nameless mode.}
\begin{lstlisting}[basicstyle=\small]
eval : global_env -> nat -> env val -> expr -> res val
\end{lstlisting}
 We will use the notation $\evaloak{\Sigma,\rho}{n}{e}$ to mean \icode{eval
   $~\Sigma~$ n $~\rho~$ e}. The \icode{global_env} parameter provides mappings
 from names of inductives to their constructors. The next three
 parameters are: ``fuel'', an evaluation environment and a
 \acornkernel{} expression. Note that ``fuel'' has a different meaning
 than ``gas'' for smart contracts: ``fuel'' is used to limit the
 recursion depth to ensure termination and not as a measure of
 computational efforts. The resulting type is \icode{res val}, where
 \icode{res} is a error monad. The errors could be either
 $\<NotEnoughFuel>$ denoting that fuel provided was not sufficient to
 complete the execution, or $\<EvalError>$ denoting that execution is
 stuck. In our Coq development $\<EvalError>$ also carries a error
 message. Values \icode{val} are defined as follows:
\begin{eqnarray*}
    v & ::= & \vConstr{I}{C}{v_1 \ldots v_n}\\ &&|~ \vClosLam{\rho}{x}{\tau}{e}\\
  &&|~ \vClosFix{\rho}{f}{x}{\tau_1}{\tau_2}{e}\\
  && |~ \vTyClos{\rho}{A}{e} ~|~\vTy{\tau} \\
\end{eqnarray*}
Note that we annotate our value with types and variable names. This is required
to match the \acornkernel{} interpreter with the MetaCoq evaluation relation (see
Section~\ref{sec:soundness}).

The outline of the most interesting parts of our interpreter is given
in Figure~\ref{fig:interpreter}. Specifically, we show the evaluation
of fixpoints and $\<case>$-expressions. The rest is quite standard and
not essentially different from other works using definitional
interpreters (e.g.~\cite{Amin:SoundnessDefInterpreters}). An essential
part of the fixpoint evaluation is how we extend the evaluation
environment $\rho$: we evaluate the body of the fixpoint $\oak{e}$ in the
environment extended with the closure of the fixpoint itself which
corresponds to the recursive call. A fixpoint binds two variables: the
outermost one corresponds to the recursive call and the innermost is a
fixpoint's argument. Thus evaluating the fixpoint's body in the
environment $v_2::v_1::\rho$ ensures that once we hit a variable
corresponding to the recursive call it will be replaced with the body
of the fixpoint again. Note that we perform an additional check if
$v_2$ is a constructor of an inductive type (possibly applied to some
arguments). This is necessary to match the evaluation of
\acornkernel{} expressions with the MetaCoq evaluation
relation. Therefore, it is not possible to pass a function as an
argument to a \acornkernel{} fixpoint. Although this sounds limiting,
the main point of the \acornkernel{} semantics is to prove the correctness
of the embedding to Coq and fixpoints in Coq are limited to structurally
recursive definitions.

When evaluating $\<case>$ expressions, our interpreter first evaluates
all the types (parameters of the inductive and the type of branches),
then the discriminee, and if the latter evaluates to a constructor,
executes a simple pattern-matching algorithm. The $\matchpatname$
function returns a branch that matches the discriminee. Next, we
evaluate the body of the selected branch in the environment extended
with the reverse list of the constructor's arguments.

Note that our interpreter evaluates all the type expressions. This is required
since we want to match \acornkernel{} evaluation with the evaluation relation of MetaCoq,
which evaluates corresponding types. The type evaluation function $\evaltypename :
\type{env}~\type{val} \times \type{\oak{type}} -> \type{res}~\type{\oak{type}}$ essentially just
substitutes values from the evaluation environment and fails if there is no
corresponding value found.

Note that we perform some validation
of expressions against the evaluation environment. This is again required for
our soundness result. We will discuss these questions in
Section~\ref{sec:soundness}.

\subsection{Translation to MetaCoq}\label{sec:translation}
We define the translations from the AST of \acornkernel{} expressions
$\type{\oak{expr}}$ and types $\type{\oak{type}}$ to MetaCoq abstract
syntax $\type{\mcoq{term}}$ by structural recursion.\footnote{By
  MetaCoq $\type{\mcoq{term}}$ we mean a corresponding inductive type
  from the PCUIC formalisation. In our Coq development, before
  unquoting a PCUIC term, we translate it into a kernel
  representation. This translation is almost one-to-one, apart from
  the application case: it is unary in PCUIC and n-ary in the kernel
  representation, but this part is quite straightforward to
  handle. Eventually, this translation will be included in the MetaCoq
  project} On Figure~\ref{fig:translation} we outline the translation
functions. We use Haskell-like notation and blue colour for
\acornkernel{} expressions and types and Coq-like notation and green
colour for MetaCoq terms. We assume that the global environment $\Sigma :
\type{global\_env}$ contains all inductive type definitions mentioned
in \acornkernel{} expressions. Under this assumption the translation
function $\expT{-}{\Sigma}$ is total.

In \acornkernel, we have two kinds of de Bruijn indices: the one for type variables
$\oak{\hat{i}}$ and the one for term variables $\oak{\overline{i}}$. In the
translation, we map both of them to the single kind of indices of MetaCoq.
Constructors are translated to MetaCoq constructors by first looking up the
corresponding constructor number in the global environment. In MetaCoq (and in
the kernel of Coq) constructors are represented as numbers corresponding to the
position of a constructor in the list of constructors for a given inductive
type. The type of global environments $\type{global\_env}$ is a list of
definitions of inductive type. The functions
{\small\begin{align*}
  \resolveindname ~:~ & \type{global\_env} \times \type{ident}\\
                     & \rightarrow \type{option}~(\type{list}~\type{constr})\\
  \resolvectorname ~:~ & \type{global\_env} \times \type{ident} \times \type{ident}\\
                      & \rightarrow \type{option}~(\mathbb{N} \times \type{constr})
\end{align*}}
\noindent
are used to look up for inductive type and their constructors. Particularly,
$\resolveindname$ returns a list of constructors for a given textual name of an
inductive type, while $\resolvectorname$ returns a position of a constructor in
the list of constructor definitions and the constructor definition
itself. Translation of a \acornkernel{} constructor $\oak{C_I}$ looks up the corresponding
constructor position in the list of constructors for the inductive $\oak{I}$,
a translated constructor in MetaCoq is basically a number (a position) annotated with the name of the
inductive type: $\mcoq{C_I}$ (and universes, but these are not relevant for us
right now). In the translation of $\oak{\<fix>}$, the type of a fixpoint is
translated into a $\Pi$-type in MetaCoq. Therefore, the indices of free variables in
the codomain type must be incremented (``lifted'') by 1, since $\Pi$-types bind a
variable. We denote such increments by $n$ as $\liftn{n}$. The other feature of the
fixpoint translation is that the body of a fixpoint in MetaCoq becomes a
lambda abstraction and since all lambda abstractions must be explicitly
annotated with a type of the domain, we provide this type. Again, we have to
lift free variables in the type, because the outermost variable index corresponds
to the body of the fixpoint itself.

By far the most complex translation case is the pattern-matching. The first
complication stems to the representation of branches for $\mcoq{\<match>}$ in
MetaCoq: the branches should be arranged in the same order as constructors in
the definition of the corresponding inductive type. In this case, there is no
need to store constructor names in each pattern. On the other hand, in \acornkernel{} we
choose more user-friendly implementation: patterns are explicitly named after
constructors and might follow in an order that is different from how the
order in the inductive type definition. For that reason, we first resolve the
inductive type from the global environment to get a list of constructors.  Then,
for each constructor in the list, we call the $\branchname$ function. As one can
see from Figure~\ref{fig:translation}, the translated branches follow the same
order as in they appear in the list of resolved constructors, i.e. $\oak{c_1 \ldots
  c_m}$.

The second difficulty arises from the pattern representation. In
MetaCoq, patterns are desugared to more basic building blocks: lambda
abstractions. Therefore, every pattern becomes an iterated lambda
term. Before we explain how the $\branchname$ function works, let us
first describe the representation of inductive types in the global
environment.  Each inductive type definition consists of a name,
number of parameters, and the list of constructors. In turn, each
constructor consists of a name and a list of argument types. Since
\acornkernel{}'s type system does not feature dependent types, it is
sufficient to store a list of arguments for each constructor instead of
a full type. Each type in the list of constructor arguments can refer
to parameters as if they were bound at the top level for each
type. For example, the type of finite maps (in the form of association
lists) would look like the following (we use concrete syntax here for
the presentation purposes):
\begin{lstlisting}[basicstyle=\footnotesize]
  data AcornMap $\#$2 = MNil [] | MCons [$\hat{1}$, $\hat{0}$, AcornMap $\hat{1}$ $\hat{0}$]
\end{lstlisting}
In the example above, \icode{AcornMap} has two parameters (the number
of parameters is specified after the type name): a type of keys, and a
type of values. The \icode{MNil} constructor does not take any
arguments, and \icode{MCons} takes a key, a value and an inhabitant of
\icode{AcornMap}. Given this representation of constructors, we
continue the explanation of the pattern-matching translation. When
pattern-matching on a parameterised inductive, the inductive is
applied to some parameters. In order to propagate these parameters to
the corresponding constructors, we have to substitute the concrete
parameters into each type in the constructor arguments list. We use
$\mcoqsubst{\id{t}}{\id{ts}}$ to denote the MetaCoq parallel substitution of a
list of terms $\id{ts}$ into a term $\id{t}$. In the $\branchname$
function, we first look up a branch body $\oak{e}$ in the list of
branches by comparing the constructor name to the pattern name. Next,
we project constructor argument types from a given
constructor. Further, since patterns become iterated lambdas in
MetaCoq, we need to provide a type for each abstracted variable. Since
Coq is dependently typed, variables bound by each lambda abstraction
might appear in the types that appear later in the term. Thus, to
avoid variable capture in types of lambda arguments, we need to lift
the translated types. This is exactly what happens in the
$\branchname$ function: the resulting MetaCoq term is an iterated
lambda abstraction and each argument type is lifted according to the
number of preceding lambda abstractions in the translation of a
pattern.

The translation of types is mostly straightforward. For the universal types, we
choose to produce a $\Pi$-type with the domain in $\mcoq{\<Set>}$. Such a choice
allows us to avoid dealing with universe levels explicitly, as it is required in
MetaCoq. The translation we present works only for a \emph{predicative} fragment
of \acornkernel{}, but \acorn's surface language supports only a prenex form of universal
types, where all the quantifiers appear at the topmost level. Therefore, this is
not a limitation from a practical point of view. Another thing to note is that
the translation of \emph{definitions} of inductive types is not shown in
Figure~\ref{fig:translation}, although it is implemented in our Coq
development. Instead of giving the full translation in the paper (which involves
subtle de Bruijn indices manipulations), let us consider an example. We continue
with the finite maps \icode{AcornMap}. The tricky bit in the translation is
to produce a correct type for each constructor considering the number of
parameters and taking into account that each $\Pi$-type binds a new
variable. Moreover, in MetaCoq the inductive type being defined becomes a
topmost variable as well.

\begin{figure*}
  \small{\begin{equation*}
    \begin{aligned}[t]
    \expT{-}{} & : \type{global\_env} \times \type{\oak{expr}} \rightarrow \type{\mcoq{term}}\\
    \expT{\oak{\overline{i}}}{\Sigma} & \defn  \mcoq{\overline{i}}\\
    \expT{\oak{\lambda x : \tau.e}}{\Sigma} & \defn \mcoq{\<fun>~ (x :}~\tyT{\oak{\tau}})~\mcoq{=>}~\expT{\oak{e}}{\Sigma}\\
    \expT{{\oak{\Lambda A.e}}}{\Sigma} & \defn {\mcoq{\<fun>~(A:\<Set>)~=>}}~\expT{\oak{e}}{\Sigma} \\
    \expT{\oak{\<let>~x:~\tau~=~e_1~\<in>~e_2}}{\Sigma} & \defn {\mcoq{\<let>~x:}}~\tyT{\oak{\tau}}~\mcoq{:=}~\expT{\oak{e_1}}{\Sigma}~\mcoq{\<in>}~\expT{\oak{e_2}}{\Sigma}\\
    \expT{\oak{e_1~e_2}}{\Sigma} & \defn \expT{\oak{e_1}}{\Sigma}~\expT{\oak{e_2}}{\Sigma}\\
    \expT{\oak{C_I}}{\Sigma} & \defn \<let> (\mcoq{C},\_) := \resolvector{\Sigma}{\oak{C}}{\oak{I}}\\
                         & \qquad \<in>~\mcoq{C}_\mcoq{I}\\
    \end{aligned}\qquad
    \begin{aligned}[t]
      \tyT{-}{} & : \type{\oak{type}} \rightarrow \type{\mcoq{term}}\\
      \tyT{\oak{\hat{i}}} &\defn \mcoq{\overline{i}}\\
      \tyT{\oak{I}} &\defn \mcoq{I}\\
      \tyT{\oak{\tau}} &\defn \tyT{\tau}\\
      \tyT{\oak{\forall A.\tau}} &\defn \mcoq{\<forall>~(A:\<Set>),}~\tyT{\oak \tau}\\
      \tyT{\oak{\tau_1~\tau_2}} & \defn \tyT{\oak{\tau_1}}~\tyT{\oak{\tau_2}}\\
      \tyT{\oak{\tau_1 \rightarrow \tau_2}} & \defn \mcoq{\<forall>}~(\mcoq{\_ :}\tyT{\oak{\tau_1}})\mcoq{,} ~\lift{\tyT{\oak{\tau_2}}}\\
  \end{aligned}
    \end{equation*}
    \begin{equation*}
      \expT{\oak{\<fix>~f~x : \tau_1 \rightarrow \tau_2 ~=~ e}}{\Sigma} \defn \mcoq{\<fix>}~(\mcoq{f}~\mcoq{:}~\mcoq{\<forall>}~(\mcoq{\_~:}~\tyT{\oak{\tau_1}})\mcoq{,}~\lift{\tyT{\oak{\tau_2}}})~\mcoq{:=}~
    \mcoq{\<fun>}~(\mcoq{x~:} \lift{\tyT{\oak{\tau_1}}})~\mcoq{=>}~\expT{e}{\Sigma}
    \end{equation*}
  \begin{flalign*}
      \left\llbracket\begin{array}{@{}ll@{}}
          \oak{\<case>}&\oak{e~:~I~\tau_1 \ldots \tau_n ~\<return>~\sigma~\<of>}\\
          & \oak{p_1\rightarrow e_1}\\
          & \ldots\\
          & \oak{p_m \rightarrow e_m}
        \end{array}\right\rrbracket_\Sigma &
      \defn
      \begin{array}{@{}ll@{}}
        \quad \mcoq{\<match>}&\expT{\oak e}{\Sigma}~\mcoq{\<as>}~\_~ \mcoq{\<in>~I~\tyT{\oak{\tau_1}} \ldots \tyT{\oak{\tau_n}}}~\mcoq{\<return>}~\lift{\tyT{\oak \sigma}}~\mcoq{\<with>}\\
        & \branch{\vect{~\oak{\tau_i}}}{\vect{\oak{p_i \rightarrow e_i}}}{\oak{c_1}}\\
        & \ldots\\
        & \branch{\vect{~\oak{\tau_i}}}{\vect{\oak{p_i \rightarrow e_i}}}{\oak{c_m}}
      \end{array}
  \end{flalign*}
  \smallskip
  \begin{flushleft}
    where
  \end{flushleft}
  \begin{equation*}
    \oak{c_1 \ldots c_m} \defn \resolveind{\Sigma}{\oak I} \qquad \vect{\oak{\tau_i}} \defn \oak{\tau_1}, \ldots,\oak{\tau_n} \qquad \tyT{\vect{~\oak{\tau_i}}} \defn \tyT{\oak{\tau_1}}, \ldots,\tyT{\oak{\tau_n}} \qquad \vect{\oak{p_i \rightarrow e_i}} \defn \oak{p_1 \rightarrow e_1}, \ldots, \oak{p_m \rightarrow e_m}
  \end{equation*}
  \begin{flalign*}
    \mathtt{branch}~:&~(\type{list}~\type{\oak{type}}) \times (\type{list}~(\type{\oak{pat}} \times \type{\oak{expr}})) \times \type{\oak{constr}} \rightarrow \type{\mcoq{term}}\\
      \branch{\vect{~\oak{\tau_i}}}{\vect{\oak{p_i \rightarrow e_i}}}{\oak c} \defn &\<let>~\oak{e}:=\find{\ctorname{\oak c}}{\vect{\oak{p_i \rightarrow e_i}}}~\<in>\\
      & \<let>~\oak{\sigma_1, \sigma_2, \ldots,\sigma_k}:=\ctorargs{\oak c}~\<in>\\
    &\mcoq{\<fun>}~\mcoq{(x_1:}~\mcoqsubst{\tyT{\oak{\sigma_1}}}{\tyT{\vect{~\oak{\tau_i}}}})~\mcoq{(x_2:}~\mcoqsubst{\tyT{\oak{\sigma_2}}}{\lift{\tyT{\vect{~\oak{\tau_i}}}}}) \ldots \mcoq{(x_k:}~\mcoqsubst{\tyT{\oak{\sigma_k}}}{\liftn{{k-1}}{\tyT{\vect{~\oak{\tau_i}}}}})~\mcoq{=>}~\expT{\oak{e}}{\Sigma}
  \end{flalign*}
  \caption{Translation to MetaCoq}\label{fig:translation}}
\end{figure*}

The resulting MetaCoq definition of \icode{AcornMap} (again, in the concrete syntax, but
with explicit indices in place of variable names) looks as follows:
\begin{lstlisting}[basicstyle=\footnotesize]
  Inductive AcornMap (A1 A2 : Set) :=
  | MNil : $\mcoq{\overline{2}}$ $\mcoq{\overline{1}}$ $\mcoq{\overline{0}}$
  | MCons : forall (_ : $\mcoq{\overline{1}}$) (_ : $\mcoq{\overline{1}}$) (_ : $\mcoq{\overline{4}}$ $\mcoq{\overline{3}}$ $\mcoq{\overline{2}}$), $\mcoq{\overline{5}}$ $\mcoq{\overline{4}}$ $\mcoq{\overline{3}}$,
\end{lstlisting}
Let us consider the type of \icode{MNil}. The index $\mcoq{\overline{2}}$ refers
to the topmost variable being the inductive type itself, index
$\mcoq{\overline{1}}$ refers to the parameter \icode{A1} and
$\mcoq{\overline{0}}$ to the parameter \icode{A2}.

Our Coq development contains the definition of \icode{AcornMap} using
the deep embedding as well as standard operations of finite
maps.\footnote{See \texttt{theories/examples/FinMap.v} file in our Coq development}
Moreover, we demonstrate how one can covert definitions, given as
a deep embedding, to regular Coq definitions by translating and unquoting
them. In the same time, one can run programs directly on deep
embedding using the interpreter (Figure~\ref{fig:interpreter}).

\subsection{Translation Soundness}\label{sec:soundness}
Since the development of the meta-theory of Coq itself is one of the
aims of MetaCoq we can use this development to show that the semantics
of \acornkernel{} agrees with its translation to MetaCoq (on
terminating programs). The idea is to compare the results of the
evaluation of \acornkernel{} expressions with the weak head
call-by-value evaluation relation of MetaCoq up to the appropriate
conversion of values. This conversion of values is a non-trivial
procedure: \acornkernel{} values contain \emph{closures}, while the
MetaCoq evaluation relation is substitution based and produces a
subset of \emph{terms} in the weak head normal form. Therefore, if we
want to eventually convert \acornkernel{} values to MetaCoq terms,
first, we need to substitute environments into the closures'
bodies. E.g. for $\vClosLam{\rho}{x}{\tau}{e}$ we need to substitute all the
values from $\rho$ into $e$. This is not possible to do directly, because
we cannot substitute values into expressions. Thus, we need to convert
all the values to expressions in the environment $\rho$. But this, in
turn, requires substituting environments in closures again. To break
this circle, we take inspiration from~\cite{Garrigue:CertifiedML} and
first define substitution functions purely on \acornkernel{}
expressions and types:
{\small\begin{align*}
\tau[-] : & ~ \type{env}~\type{\oak{expr}} -> \type{option}~\type{\oak{type}}\\
e[-] : & ~\type{env}~\type{\oak{expr}} -> \type{option}~\type{\oak{expr}}
\end{align*}}
These functions implement parallel substitution of the environment represented as
a list of expressions. Unfortunately, these functions are partial, due to the
fact that we use one environment for term-level values and for
type-level-values. We can make this function total by imposing a well-formedness
condition. With these substitution operations we now can define a conversion
procedure from \acornkernel{} values back to expressions:
{\small\begin{align*}
  \fromval{\vClosLam{\rho}{x}{\tau}{\oak{e}}} \defn & \<let>~\rho' :=\map{\fromvalname}{\rho}~\<in>\\
                                              & (\oak{\lambda x : \tau.e})[\rho']\\
  \fromval{\vClosFix{\rho}{f}{x}{\tau_1}{\tau_2}{\oak{e}}} \defn &\<let>~\rho' :=\map{\fromvalname}{\rho}~\<in>\\
                                                        & (\oak{\<fix>~f~x : \tau_1 \rightarrow \tau_2 ~=~ e})[\rho']\\
  \fromval{\vConstr{I}{C}{v_1 \ldots v_n}} \defn & C_I ~\fromval{v_1} \ldots \fromval{v_n}\\
  \fromval{\vTyClos{\rho}{A}{\oak{e}}} \defn & \<let> \rho' :=\map{\fromvalname}{\rho}~\<in>\\
  & ({\oak{\Lambda A.e}})[\rho']\\
  \fromval{\vTy{\oak{\tau}}} \defn & \oak{\tau}
\end{align*}}
Once we have a way of converting values to expressions, we can use the
translation function $\expT{-}{}$ to produce MetaCoq terms. This gives us a
direct way of comparing the evaluation results. Before we state the soundness
theorem, we give an overview of some important lemmas forming a core of the
proof. First of all, let us mention certain well-formedness conditions for the
environments involved in our definitions. It is very important to carefully set
up these conditions before approaching the soundness proof.
\begin{definition}\label{defn:wf}
  For a global environment $\Sigma:\type{global\_env}$, evaluation environment $\rho:\type{env}~\type{\oak{expr}}$ and \acornkernel{} value $v:\type{val}$ we say that
  \begin{enumerate}[label=(WF.\roman*),ref=(WF.\roman*)]
    \item\label{enum:wf:genv} $\Sigma$ is well-formed if for all definitions of inductive types, each
      constructor type is closed for the given number of parameters of the
      inductive type. E.g. if an inductive type has $n$ parameters, then the type of
      each constructor has at most $n$ free variables.
    \item\label{enum:wf:env} $\rho$ is well-formed wrt. an expression $e$ when for any type variables
      mentioned in $e$, if there is a corresponding expression in $\rho$ it corresponds to a type.
    \item\label{enum:wf:val} a value $v$ is well-formed if all the expressions and types in the
      closures are appropriately closed wrt. corresponding environments in
      closures and $\rho$ is well-formed in the sense of~\ref{enum:wf:env}.
      E.g. for $\vClosLam{\rho}{x}{\tau}{e}$ we have: $\rho$ contains only well-formed
      values, $e$ has at most $|\rho|+1$ free variables ($|\rho|$ is the size of $\rho$), $\tau$ is closed type value,
      and $\rho$ is well-formed wrt. $e$.
      \sloppy
      Additionally, for $\vConstr{I}{C}{\id{args}}$, we require that $\resolvector{\Sigma}{I}{C}$ returns some value.
  \end{enumerate}
\end{definition}

Now, we will state several lemmas crucial for the soundness proof. We
will emphasise the use of the conditions
\textit{\ref{enum:wf:genv}},\textit{\ref{enum:wf:env}} and
\textit{\ref{enum:wf:val}} throughout these lemmas.  We will use the
following additional notations: $\mcoqsubst{\id{t}}{\id{ts}}$ for the
MetaCoq parallel substitution as in the translation, $\expT{\rho}{\Sigma}$ for
translation of all the expressions in $\rho$ from \acornkernel{} to
MetaCoq and $|\rho|$ for the environment size.

\begin{lemma}[Environment substitution]\label{lem:env-subst}
  For any \acornkernel{} expression $e$, well-formed global environment $\Sigma$
  \ref{enum:wf:genv}, well-formed environment $\rho$ wrt. $e$ \ref{enum:wf:env}, such that
  all the expressions in $\rho$ are closed the following holds
  \[ \expT{e[\rho]}{\Sigma} = (\expT{e}{\Sigma})\big\{\expT{\rho}{\Sigma}\big\} \]
\end{lemma}

Lemma~\ref{lem:env-subst} says that environment substitution commutes with the
translation. The well-formedness condition on $\rho$ ensures that the environment
substitution functions are total.

\begin{lemma}[Well-formed values]\label{lem:eval-val-ok}
  For any \acornkernel{} expression $e$, such that $e$ has at most $|\rho|$ free variables,
  number of steps $n$, well-formed evaluation environment $\rho$, such that all the
  values in $\rho$ are well-formed, if evaluation of $e$ terminates with some value
  $v$, i.e. $\evaloak{\Sigma,\rho}{n}{e} = \<Ok>~v$, then the value $v$ is well-formed (\ref{enum:wf:val}).
\end{lemma}

In the interpreter in Figure~\ref{fig:interpreter} we perform certain dynamic
checks denoted by $\validatename$ and $\validatebranchesname$. These checks
ensure that the condition~\textit{\ref{enum:wf:val}} is satisfied for values
produces by the interpreter. For well-typed expressions, this condition would be
automatically satisfied, but currently, we focus on dynamic semantics.

We use the interpreter for \acornkernel{} expressions and call-by-value evaluation relation
of MetaCoq to state the soundness theorem. The MetaCoq evaluation relation is a
subrelation of the transitive reflexive closure of the one-step reduction relation
and designed to represent the evaluation of ML languages at Coq level.
\begin{theorem}[Soundness]\label{thm:soundness}
  For any \acornkernel{} expression $e$, number of steps $n$, well-formed global
  environment $\Sigma$, evaluation environment $\rho$, such that all the values in $\rho$
  are well-formed and $e[\rho]$ is closed, if $\evaloak{\Sigma,\rho}{n}{e} = \<Ok>~v$, for
  some value $v$, then $\expT{e[\rho]}{\Sigma} \Downarrow \expT{\fromval{v}}{\Sigma}$, where $- \Downarrow -$
  is the call-by-value evaluation relation of MetaCoq.
\end{theorem}
\begin{proof}
  By induction on the number of steps of the interpreter $n$. The base case is trivial,
  since we assume that the interpreter terminates. In the inductive step, the
  proof proceeds by case analysis on $e$ using Lemma~\ref{lem:env-subst} in the
  cases involving substitution (e.g. cases for let, application and
  $\<case>$-expressions) and using Lemma~\ref{lem:eval-val-ok} to obtain premises
  that all the values in $\rho$ are well-formed required for applying induction
  hypotheses.
\end{proof}

\begin{corollary}[Soundness for closed expressions]\label{thm:soundness-closed}
  For any closed \acornkernel{} expression $e$, number of steps $n$, well-formed global
  environment $\Sigma$, if $\evaloak{\Sigma,[]}{n}{e} = \<Ok>~v$,
  then $\expT{e}{\Sigma} \Downarrow \expT{\fromval{v}}{\Sigma}$.
\end{corollary}
\begin{proof}
  By Theorem~\ref{thm:soundness}, using the fact that the empty evaluation
  environment is trivially well-formed and the fact that substituting the empty
  environment does not change $e$.
\end{proof}

We can see our translation to MetaCoq as some form of denotational semantics. With
this view, we can obtain the adequacy result.

\begin{theorem}[Adequacy for terminating programs]\label{thm:adequacy}
  For any closed \acornkernel{} expression $e$, well-formed global environment $\Sigma$, if
  evaluation of $e$ terminates with value $v$ in $n$ steps and
  $\expT{e}{\Sigma} \Downarrow t$ for some term $t$, then $t = \expT{\fromval{v}}{\Sigma}$.
\end{theorem}
\begin{proof}
  By Corollary~\ref{thm:soundness-closed}, using the fact that the MetaCoq
  evaluation relation is deterministic.
\end{proof}

Theorem~\ref{thm:adequacy} readily allows, for example, transferring
program equivalence from MetaCoq derivations to the \acornkernel{}
interpreter, provided that the values are not higher-order (i.e. not
closures).

In general, we conjecture that one can show adequacy for any program
for which there exists a derivation of the MetaCoq big-step evaluation
relation. Such a theorem should be stated for well-typed
\acornkernel{} expressions. Moreover, transferring properties proved
for Coq functions to the corresponding evaluations using the
interpreter in general also requires resorting to the typing
argument. Currently, we do not formalise static semantics and leave
these points as future work.



We assume that the \icode{unquote} functionality of MetaCoq is
implemented correctly. From that perspective, \icode{unquote} becomes
part of the trusted computing base, but we would like to emphasise
that one of the goals of the MetaCoq project is to implement the
actual kernel of Coq in Coq itself. The current MetaCoq data type
\icode{term} corresponds directly to the \icode{constr} data type from
Coq's kernel. Therefore, \icode{unquote} is a straightforward
one-to-one mapping of MetaCoq data types to the corresponding OCaml
data types of Coq's kernel. This is in contrast to projects like
hs-to-coq and coq-of-ocaml for which the whole
translation has to be trusted.

We developed a full formalisation of theorems and lemmas presented in
this section in our framework in Coq. We do not use any extra axioms
throughout our development, but Theorem~\ref{thm:adequacy} uses the
determinism of the MetaCoq evaluation relation and the proof of this
fact is currently under development in the MetaCoq project. Being able
to relate the semantics of \acornkernel{} to the semantics of Coq
through Coq's meta-theory formalisation gives strong guarantees that
our shallow embedding reflects the actual behaviour of
\acornkernel{}. The described approach provides a more principled way
of embedding of functional languages than the source-to-source based
approaches. Moreover, the translation involves manipulation of de
Bruijn indices, which is often quite hard to get right. Various
mistakes in non-trivial places were discovered and fixed in the
course of the formalisation.

\section{The Crowdfunding Contract}\label{sec:crowdfunding}
As an example of our approach, we consider verification of some
properties of a crowdfunding contract
(Figure~\ref{fig:crowdfunding}). Such a contract allows arbitrary
users to donate money within a deadline. If the crowdfunding goal is
reached, the owner can withdraw the total amount from the account
after the deadline has passed. Also, users can withdraw their
donations after the deadline if the goal has not been
reached. Contracts like this are standard applications of smart
contracts and appear in a number of tutorials.\footnote{The idea of a
  crowdfunding contract appears under different names: crowdsale,
  Kickstarter-like contract, ICO contract, etc. Many Ethereum-related
  resources contain variations of this idea in tutorials (including
  Solidity and Vyper documentation).  A simplified version of a
  crowdfunding contract is also available for Liquidity:
  \url{https://github.com/postables/Tezos-Developer-Resources/blob/master/Examples/Crowdfund/Basic.ml}}
We follow the example of Scilla~\cite{Sergey:Scilla} and adopt a
variant of a crowdfunding contract as a good instance to demonstrate
our verification techniques.

\lstset{basicstyle=\scriptsize}
\begin{figure*}
  \begin{lstlisting}[multicols=2]
(* Defining AST using customised notations                                *)
(* Brackets [\ \] delimit the scope of global definitions, *)
(*             [| |] the scope of programs,                               *)
(*             [! !] the scope of types                                   *)
(* Local state *)
Definition state_syn : global_dec :=
[\ record State :=
     mkState { balance : Money ;
               donations : Map;
               owner : Address;
               deadline : Nat;
               done : Bool;
               goal : Money } \].
Make Inductive (trans_global_dec state_syn).

(* Messages. Constructors do not take any arguments.*)
Definition msg_syn :=
  [\ data Msg =
         Donate [_]
       | GetFunds [_]
       | Claim [_] \].

Make Inductive (trans_global_dec msg_syn).

(* Abbreviations for types of the blockchain infrastructure *)
Notation SActionBody := "SimpleActionBody".
Notation SCtx := "SimpleContractCallContext".
Notation SChain := "SimpleChain".

(* An abbreviation for the return type *)
Notation "'Result'" := [! "prod" State ("list" "SimpleActionBody") !]
                         (in custom type at level 2).

(* Initialisation function *)
Definition crowdfunding_init : expr :=
[| \c : SCtx => \dl : Nat => \g : Money =>
  mkState 0z MNil dl (ctx_from c) False g |].

Make Definition init :=
  (expr_to_term Σ' (indexify nil crowdfunding_init)).

(* The main functionality *)
Definition crowdfunding : expr :=
[| \chain : SChain =>  \c : SCtx => \m : Msg => \s : State =>
    let bal : Money := balance s in
    let now : Nat := cur_time chain in
    let tx_amount : Money := amount c in
    let sender : Address := ctx_from c in
    let own : Address := owner s in
    let accs : Map := donations s in
    case m : Msg return Maybe Result of
       | GetFunds ->
        if (own == sender) && (deadline s < now) && (goal s <= bal) then
          Just (Pair (mkState 0z accs own (deadline s) True (goal s))
                       [Transfer bal sender])
        else Nothing : Maybe Result
      | Donate -> if now <= deadline s then
        (case (mfind accs sender) : Maybe Money return Maybe Result of
          | Just v ->
            let newmap : Map := madd sender (v + tx_amount) accs in
            Just (Pair (mkState (tx_amount + bal) newmap own
                                  (deadline s) (done s) (goal s)) Nil)
          | Nothing ->
            let newmap : Map := madd sender tx_amount accs in
            Just (Pair (mkState (tx_amount + bal) newmap own
                                  (deadline s) (done s) (goal s)) Nil))
          else Nothing : Maybe Result
      | Claim ->
        if (deadline s < now) && (bal < goal s) && ($\sim$ done s) then
        (case (mfind accs sender) : Maybe Money return Maybe Result of
         | Just v -> let newmap : Map := madd sender 0z accs in
           Just (Pair(mkState (bal$-$v) newmap own (deadline s) (done s) (goal s))
                       [Transfer v sender])
          | Nothing -> Nothing)
         else Nothing : Maybe Result |].


Make Definition receive :=
    (expr_to_term Σ' (indexify nil crowdfunding)).




  \end{lstlisting}
  \caption{The crowdfunding contract}\label{fig:crowdfunding}
\end{figure*}

\lstset{basicstyle=\footnotesize}

We extensively use a new feature of Coq called ``custom entries'' to provide a
convenient notation for our deep embedding.\footnote{Custom entries are
  available starting from Coq 8.9.0.} The program texts in
Figure~\ref{fig:crowdfunding} written inside the special brackets
\icode{[\\ ... \\]}, \icode{[| ... |]}, and \coqe{[! ... !]} are parsed
according to the custom notation rules. For example, without using notations the
definition of \icode{action_syn} looks as follows:
\begin{lstlisting}
gdInd Action 0 [("Transfer", [(nAnon, tyInd "nat");
                  (nAnon, tyInd "nat")]);("Empty", [])] false.
\end{lstlisting}
This AST might be printed directly from the smart contract AST by a
simple procedure (as we use in Section~\ref{sec:oak-stdlib}). We start by
defining the required data structures such as \icode{State} and \icode{Msg}
meaning contract state and messages accepted by this contract. We pre-generate
string constants for corresponding names of inductive types, constructors,
etc. using the MetaCoq template monad.\footnote{The template monad is a part of
  the MetaCoq infrastructure. It allows for interacting with Coq's global
  environment: reading data about existing definitions, adding new definitions,
  quoting/unquoting definitions, etc.} This allows for more readable
presentation using our notation mechanism. Currently, we use the \icode{nat}
type of Coq to represent account addresses and currency. Eventually, these types
will be replaced with corresponding formalisations of these primitive types. We
also use abbreviations for the result type and for certain types from the
blockchain infrastructure, which we are going to explain later.

The \icode{trans_global_dec : global_dec -> mutual_inductive_entry}\newline
function takes the syntax of the data type declarations and produces
an element of \icode{mutual_inductive_entry} --- a MetaCoq
representation for inductive types.  For each of our deeply embedded
data type definitions, we produce corresponding definitions of
inductive types in Coq by using the \icode{Make Inductive} command of
MetaCoq that ``unquotes'' given instances of the
\icode{mutual_inductive_entry} type. Similar notation mechanism is
used to write programs using the deep embedding. The definition of
\icode{crowdfunding} represents a syntax of the crowdfunding
contract. We translate the crowdfunding contract's AST into a MetaCoq
AST using the \icode{expr_to_term : global_env -> expr -> term}
function (corresponding to $\expT{-}{\Sigma}$ in
Figure~\ref{fig:translation}). Here, \icode{global_env} is a global
environment containing declarations of inductive types used in the
function definition, \icode{expr} is a type of \acornkernel{}
expressions, and \icode{term} is a type of MetaCoq terms.  Before
translating the \acornkernel{} expressions, we apply the
\icode{indexify} function that converts named variables into de Bruijn
indices. The result of these transformations is unquoted with the
\icode{Make Definition} command. After unquoting the translated
definitions, they are added to Coq's global environment and available
for using as any other definitions. The crowdfunding contract consists
of two functions:
\begin{lstlisting}
  init : SimpleContractCallContext -> nat -> Z -> State_coq

  receive : SimpleChain -> SimpleContractCallContext
            -> Msg_coq -> State_coq
            -> option (State_coq × list SimpleActionBody)
\end{lstlisting}
\noindent
Here \icode{SimpleChain} is a ``contract's view'' of a blockchain allowing for
accessing, among other parameters, current slot number (used as a current time);
\icode{SimpleContractCallContext} is a contract call context containing
transferred amount, sender's address and other information available for
inspection during the contract call. The type names with the ``coq'' suffix
correspond to the unquoted data types from the Figure~\ref{fig:crowdfunding}.

The \icode{init} function sets up an initial state for a given deadline and goal. The
\icode{receive} function corresponds to a transition from a current state of
the contract to a new state. We will provide more details about the execution
model in Section~\ref{sec:exec-framework}. In the current section we focus on
functional correctness properties using pre- and post-conditions. Similarly
to~\cite{Sergey:Scilla}, we prove a number of properties of the contract using
the shallow embedding:
\begin{enumerate}[label=(P.\roman*),ref=(P.\roman*)]
\item\label{enum:cfprop:sum} the contract \icode{receive} function preserves the following invariant (unless the ``done'' flag is set to \icode{true}):
  the sum of individual contributions is equal to the balance recorded in the
  contract's state;
\item the donations can be paid back to the backers if the goal is not reached within a deadline;
\item donations are recorded correctly in the contract's state;
\item backers cannot claim their contributions if the campaign has succeeded.
\end{enumerate}
The lemma corresponding to the property~\ref{enum:cfprop:sum} is given below.
\begin{lstlisting}
  Lemma contract_state_consistent BC CallCtx msg :
    {{ consistent_balance }}
      receive BC CallCtx msg
    {{ fun fin txs => consistent_balance fin }}.
\end{lstlisting}

In the example above we use the Hoare triple notation
\icode{\{\{P\}\}c\{\{Q\}\}} to state pre- and post-conditions for the
state before and after the contract call. The post-condition also
allows for stating properties of outgoing transactions. We define
\icode{consistent_balance} as follows:
\begin{lstlisting}
Definition consistent_balance (lstate : State_coq) :=
  $\sim$ lstate.(done_coq) ->
  sum_map (donations_coq lstate) = balance_coq lstate.
\end{lstlisting}
\noindent i.e. the contract balance is consistent if before the contract
is marked as ``done'' the sum of individual contributions is equal to the
balance.

Given the definitions above, one can read the lemma in the following way:
if the balance was consistent in some initial state, then execution of
the \icode{receive} method gives a new state in which the balance is again
consistent. Note that the \icode{receive} is a ``regular'' Coq
function and it is a shallow embedding of the corresponding
\icode{crowdfunding} definition (Figure \ref{fig:crowdfunding}) produced
automatically by our translation.

\section{Verifying Standard Library Functions}\label{sec:oak-stdlib}

In our Coq development, we show how one can verify \acorn{} library
code by proving \acornkernel{} functions (obtained by printing the
\acorn{} AST as \acornkernel{} AST) equivalent to the corresponding
functions from the standard library of Coq. In
particular, we provide an example of such a procedure for certain
functions on lists. The similar approach is mentioned as a strong side
of Coq in comparison to Liquid Haskell~\cite{Vazou:LiqHaskellCoq}. In
general, our framework will be applicable for verification of standard
libraries of various functional languages (not even necessarily
languages for smart contracts) since data types such as lists, trees,
finite maps, etc. are ubiquitous in functional programming. In
addition, \acornkernel{} is essentially a pure fragment of various
functional general-purpose languages (ML-family, Elm) and smart
contract languages (Liquidity, Simplicity, Sophia\footnote{A
  functional smart contract language based on ReasonML:
  https://dev.aepps.com/aepp-sdk-docs/Sophia.html}) making it a good
target for integration.

Figure~\ref{fig:oak-lists} shows how we obtain the shallow embedding from the
\acornkernel{} \texttt{List} module of the \acorn{} standard library. We start with the concrete
syntax (Figure~\ref{subfig:oak-code}). Next, from the \acorn{} parser, which is
a part of the Concordium infrastructure, we obtain an AST. This AST is printed
to obtain a Coq representation (which we call \acornkernel) using a simple
printing procedure implemented in Haskell (Figure~\ref{subfig:oak-ast}). From
the module AST in Coq we produce a shallow embedding by using the translation
described in Section~\ref{sec:translation} and the \icode{TemplateMonad} of
MetaCoq to unquote all the definitions in the module. We use
\icode{translateData} and \icode{translateDefs} function to translate and
unquote all the definitions of data types and functions respectively. The
\icode{translateDefs} is defined as follows:

\begin{lstlisting}
Fixpoint translateDefs (Σ : global_env) (es : list (string * expr))
  : TemplateMonad unit:=
  match es with
  | [] => tmPrint "Done."
  | (name, e) :: es' =>
    coq_expr <- tmEval all (expr_to_term Σ (reindexify 0 e)) ;;
    print_nf ("Unquoted: " ++ name);;
    tmMkDefinition name coq_expr;;
    translateDefs Σ es'
  end.
\end{lstlisting}
\noindent
The \icode{expr_to_term Σ} corresponds to $\expT{-}{\Sigma}$ from
Figure~\ref{fig:translation}. In the \acorn{} AST received from the parser, the
index spaces for type variables and term variables are separated. We merge them
into a single index space with the \icode{reindexify} function. Finally,
\icode{tmMkDefinition} unquotes the translated MetaCoq term and adds it to the
Coq environment. After that, we can interact with the unquoted definitions as if
they were written by hand.

On the shallow embedding we establish an isomorphism between \acorn{} lists and Coq
lists by defining two functions \icode{to_acorn} and \icode{from_acorn} composing to
identity. We can state the following: \icode{foldr f a l = fold_right f a
  (from_acorn l)} where \icode{foldr} is an \acorn{} function and \icode{fold_right}
comes from the standard library of Coq. Similarly for the list
concatenation. Now, we can transfer properties of these functions without the
need of reproving:
\begin{lstlisting}
Lemma foldr_concat (A B : Set) (f : A -> B -> B)
                      (l l' : AcornList A) (i : B) :
  foldr f i (concat l l') = foldr f (foldr f i l') l.
Proof. autorewrite with hints;apply fold_right_app. Qed.
\end{lstlisting}
Currently, we use \icode{autorewrite} to automate such proofs, but in the
future, we consider using more principled techniques
like~\cite{Tabareau:2018:EFU}.

\begin{figure*}
  \begin{subfigure}[b]{0.3\textwidth}
    \begin{lstlisting}
module ListBase where

import Prod
import Bool

data List a = Nil [] | Cons [a, (List a)]

definition foldr a b (f :: a -> b -> b)
   (initVal :: b) =
   letrec go (xs :: List a) :: b =
          case xs of
             Nil -> initVal
             Cons x xs' -> f x (go xs')
             in go
 $\ldots$
    \end{lstlisting}
    \subcaption{A fragment of \acorn{} code}\label{subfig:oak-code}
  \end{subfigure}
\begin{subfigure}[b]{0.3\textwidth}
  \begin{lstlisting}
Definition Data :=
[gdInd "List" 1 [("Nil_coq", []);
    ("Cons_coq", [(None, tyRel 0);
      (None, (tyApp (tyInd "List")
      (tyRel 0)))])] false].

Definition Functions :=
[("foldr", eTyLam "A" (eTyLam "A"
  (eLambda "x" (tyArr (tyRel 1)
  (tyArr (tyRel 0) (tyRel 0)))
  (eLambda "x" (tyRel 0)
  (eLetIn "f" (eFix "rec" "x" $\ldots$)))));
  $\ldots$
  \end{lstlisting}
  \subcaption{A fragment of \acornkernel{} AST (deep embedding)}\label{subfig:oak-ast}
\end{subfigure}
\begin{subfigure}[b]{0.3\textwidth}
  \begin{lstlisting}
Import AcornProd.
Import AcornBool.

Run TemplateProgram (translateData Data).
Definition gEnv := StdLib.Σ ++ Data ++ AcornBool.Data ++ AcornProd.Data.
Run TemplateProgram (translateDefs gEnv Functions).
Print foldr.
(* fun (A A0 : Set)(x : A -> A0 -> A0)
         (x0 : A0) =>
   fix rec (x1 : List A) : A0 :=
      match x1 with
      | @Nil_coq _ => x0
      | @Cons_coq _ x2 x3 =>
          x x2 (rec x3)
      end *)
  \end{lstlisting}
  \subcaption{Shallow embedding}\label{subfig:oak-shallow-embedding}
\end{subfigure}
\caption{Translating \acorn{} list functions to Coq}\label{fig:oak-lists}
\end{figure*}

\section{The Execution Framework}\label{sec:exec-framework}
In the context of blockchains smart contracts are small programs that are
published to the nodes of the system and associated with some address. Calls
happen when a transaction is made to this address and nodes execute the program
when seeing such a transaction, which additionally can contain input parameters
to the program. Smart contracts typically survive across calls and are thus
long-lived stateful objects that end up being executed multiple times. In
addition smart contracts interact with the blockchain in various ways, for
example by making calls to other smart contracts or by transferring money owned
by the smart contract into other accounts. The blockchain software thus keeps
track of extra information about the smart contract: its monetary balance and
the local state that the particular smart contract wishes to persist between
calls.

For these reasons it does not suffice to prove only functional correctness
properties if one wants to achieve strong guarantees. As an example, we would
like to prove that the crowdfunding contract in Figure~\ref{fig:crowdfunding}
has enough money when it attempts to pay back the funders. Here simple
functional correctness is uninteresting since the crowdfunding contract just
takes its own balance as an input to the function. Instead we would like a more
comprehensive model of a blockchain capturing the semantics of transactions that
affect state such as the balance. Such a model is given in \cite{Interactions}
which provides a Coq formalization of a small-step operational semantics of
blockchain execution. In this framework one can reason about multiple deployed
and interacting contracts with the system tracking the balance and local state
of each contract. This small-step semantics is used to define a trace type, and
we can specify our stronger safety properties as properties that hold for any
blockchain state reachable through a trace. Here we differentiate between
blockchain state, which is the full state of the entire blockchain, and local
state, which simply is the state that some particular contract wants the
blockchain to persist for it. Logically the local state of a smart contract is
one part of the blockchain state, but the blockchain state also contains
information like the balance of each account in the system. By using the traces
it is furthermore possible to define temporal properties.

As a running example, we continue with the crowdfunding contract from
Section~\ref{sec:crowdfunding}. We demonstrate how the
property~\ref{enum:cfprop:sum} (Section~\ref{sec:crowdfunding}) can be
transformed into a safety property and prove that the balance stored in the
internal state of the contract is always less or equal to the actual balance
recorded in the blockchain state.

In \cite{Interactions} contracts are represented as two functions \icode{init}
and \icode{receive}. The \icode{init} function is called when the contract is
deployed and allows the contract to establish its initial local state, while the
\icode{receive} function is called after deployment when transactions are sent
to the contract. These functions are provided with information about the
blockchain and the \icode{receive} function allows the contract to interact with
the blockchain more actively, such as by making transactions to other accounts.
The \icode{init} and \icode{receive} functions are partial allowing contracts to
communicate that they were called with invalid parameters.

To be able to state lemmas, we need to do some extra work by wrapping
the functions \icode{init} and \icode{receive} from
Figure~\ref{fig:crowdfunding} to have compatible signatures. In the
execution framework, these functions take \icode{Chain} and
\icode{ContractCallContext} types that are generalized over the type
of addresses used. For our particular contract we use ``plain''
inductive types \icode{SimpleChain} and
\icode{SimpleContractCallContext} for which addresses are always
natural numbers. Thus we instantiate the execution framework's
addresses as natural numbers and then define conversion functions
between the types from the execution framework and the simple variants
used by the crowdfunding contract. In effect the resulting functions
have the following types:
\begin{lstlisting}
 wrapped_init : Chain -> ContractCallContext
                    -> Setup
                    -> option State_coq

wrapped_receive : Chain -> ContractCallContext
                    -> State_coq -> option Msg_coq
                    -> option (State_coq * list ActionBody)
\end{lstlisting}
\noindent
The \icode{Setup} type here is just for packing together parameters to the init
function: deadline and goal. Another requirement for using the execution
framework is to provide instances for serialisation/\-deserialisation of the
local state (\icode{State_coq}) and messages (\icode{Msg_coq}). These instances
can be automatically generated by the execution framework for the simple
non-recursive data types used in the crowdfunding contract. With these instances
in place, we can put together \icode{init} and \icode{receive} to define a
contract:
\begin{lstlisting}
Definition cf_contract : Contract Setup Msg_coq State_coq :=
  build_contract wrapped_init init_proper wrapped_receive receive_proper.
\end{lstlisting}

Here \icode{init_proper} and \icode{receive_proper} are proofs showing that
the wrapped functions respect extensional equality on the input parameters.

Now we are ready to formulate a safety property of the crowdfunding contract:
\begin{lstlisting}
Lemma cf_balance_consistent bstate cf_addr lstate :
  reachable bstate ->
  env_contracts bstate cf_addr = Some (cf_contract : WeakContract) ->
  cf_state bstate cf_addr = Some lstate ->
  consistent_balance lstate.
\end{lstlisting}

One can read this lemma as follows: for any reachable blockchain state
\icode{bstate}, for an instance of the crowdfunding contract deployed at the
address \icode{cf_addr} with some local state \icode{lstate}, the balance
recorded in the local state is consistent with the map of individual
contributions.

Next, we state and prove the following lemma:\footnote{The lemma
  presented in the paper is a corollary of a more general theorem
  which generalises over any outstanding actions, somewhat similarly
  to a theorem about the Congress contract in~\cite{Interactions}. We
  leave out the details here, but full proofs of all the lemmas from this
  section is available in our Coq development.}
\begin{lstlisting}
Lemma cf_backed_after_block {ChainBuilder : ChainBuilderType}
          prev hd acts new cf_addr lstate :
builder_add_block prev hd acts = Some new ->
env_contracts new cf_addr = Some (cf_contract : WeakContract) ->
cf_state new cf_addr = Some lstate ->
(account_balance (env_chain new) cf_addr>=balance_coq lstate)%Z.
\end{lstlisting}
\noindent
This lemma says that after adding a new block,\footnote{For our purpose a block
can be thought of as just a list of transactions.} for any instance of the
crowdfunding contract deployed at the address \icode{cf_addr}, the actual
contract balance (the one recorded in the blockchain state) is greater or equal
to the balance recorded in the local state (used in all operations of the
contract's ``logic''). In essence our integration with the execution framework
allows us to prove that the contract tracks its own balance correctly even in
the face of potential reentrancy or nontrivial interactions with other contracts
and accounts. This is unlike previous work such as~\cite{Sergey:Scilla} which
only considers a single contract in isolation, or~\cite{Mi-Cho-Coq}, which
focuses on functional correctness properties. This lemma in combination with
\icode{cf_balance_consistent} also gives a formal argument that the contract's
account has enough money to cover all individual contributions.

Having our shallow embedding integrated with the execution framework,
one can imagine many more important safety and temporal properties
of the contract. For example, if the contract is funded, there will ever
be at most one outgoing transaction (to the owner), or if the contract
is not funded and the deadline has passed users can get their
contributions back. We leave this for future work.

\section{Related Work}
In this work, we focus on modern smart contract languages based on a
functional programming paradigm. In many cases, various small errors
in smart contracts can be ruled out by the type systems of these
languages. Capturing more serious errors requires employing such
techniques as deductive verification (for verification of concrete
contracts) and formalisation of meta-theory (e.g. to ensure the
soundness of type systems). The formalisation of the Simplicity
language~\cite{O'Connor:Simplicity} features well-developed
meta-theory, including a formalisation of the operational semantics of
the Bit Machine allowing for reasoning about computational
resources. But this formalisation does not focus on using the shallow
embedding for proving properties of smart contracts and Simplicity is
a low-level language in comparison to \acornkernel{} (Simplicity
is a non-Turing-complete first-order language and does not feature
algebraic data types).

The work on Scilla~\cite{Sergey:Scilla} focuses on verification of
concrete smart contracts in Coq. It considers a crowdfunding smart
contract example translated into Coq by hand, and the correspondence
to Scilla's meta-theory is not clear. A recent paper on
Scilla~\cite{Sergey:ScillaOOPSLA} gives a formal definition of the
language semantics, but does not feature mechanised proofs.

The formalisation of the Plutus Core
language~\cite{Chapman:PlutusCore} covers the meta-theory of System
$F_\omega^\mu$ --- a polymorphic lambda calculus with higher-order kinds and
iso-recursive types. The main difference with \acornkernel{} is the
absence of ``native'' inductive types.  Another work on
Plutus~\cite{Jones:UnravelingRecursion} shows how to compile inductive
types into System $F_\omega^\mu$. The compilation procedure is
type-preserving (which follows from the intrinsic encoding used in the
formalisation); computational soundness is left as future work.

The Michelson language formalisation~\cite{Mi-Cho-Coq} defines an
intrinsic encoding of the language expression along with its
interpreter in Coq. The Coq development uses a weakest precondition
calculus on deeply-embedded Michelson expressions for smart contract
verification.  The semantics of the Liquidity language given
in~\cite{Liquidity} provides the rules for compilation to Michelson,
but there is no corresponding formalisation in a proof assistant. The
Sophia smart contract language also belongs to the family of
functional smart contract languages based on ReasonML, but there is no
corresponding formalisation available.

There are several well-developed formalisations of variations of
System $F$
~\cite{Kovacs:SystemFw,Stucki:SystemFwIso,Garrigue:CertifiedML,Weirich:SystemFC}.
Among these,~\cite{Garrigue:CertifiedML} features an interpreter and
helped us to shape our representation of the language in Coq.

The purpose of the \icode{Program} tactic in
Coq~\cite{Parent:Program,Sozeau:SubsetCoercions} is to embed a
functional language into Coq allowing for writing specifications in
types. Our work can be seen as the first step towards making a
certified version of such a tactic.

Finally, meta-programming
techniques have also been shown to be useful in the dependently typed
setting in other proof assistants: Agda's Reflection
library~\cite{VanDerWalt:AgdaReflection}, meta-programming frameworks
in Lean~\cite{Ebner:LeanMetaprogramming} and
Idris~\cite{Christiansen:IdrisReflection} employ techniques similar to
MetaCoq.

\section{Conclusion and Future Work}
We have presented the ConCert smart contract verification
framework. An important feature of our approach is the ability to both
develop a meta-theory of a smart contract language and to conveniently
reason about particular programs (smart contracts). We proved
soundness theorems relating meta-theory of the smart contract language
with the embedding. Such an option is usually not available for
source-to-source translations. We applied our approach to the
development of an embedding of the \acornkernel{} smart contract
language and provided a verification example of a crowdfunding
contract starting from the contract's AST. We also demonstrated how
our framework can be used to verify ``standard library'' functions
common to functional programming languages by proving them equivalent
to Coq standard library functions. Moreover, we integrated the shallow
embedding with the smart contract execution framework which gives
access to properties related to the interaction of smart contracts with
the underlying blockchain and with each other. Our work addresses a
number of future work points from a recent Scilla
paper~\cite{Sergey:ScillaOOPSLA}: we provide a shallow embedding,
integrate it with a reasoning framework for safety and temporal
properties and implement a reference evaluator in Coq. To the best of
our knowledge, the ConCert framework together with the execution model
described in Section~\ref{sec:exec-framework} is the first development
allowing to verify functional smart contracts against a blockchain
execution model by automatic translation to the shallow embedding.

Our framework is general enough to be applied to other functional
smart contract languages. We consider benchmarking our development by
developing ``backends'' for translation of other languages
(e.g. Liquidity, Simplicity, etc.).  Extending the formalisation of
the \acornkernel{} language meta-theory is also among our goals for
the framework. An important bit of the meta-theory is the cost semantics
allowing for reasoning about ``gas''. We would like to give a cost
semantics for the deep embedding and explore how it can be extended on
the shallow embedding. Another line of future work is extending
\acorn{} with specification annotations similar to the Russell
language~\cite{Sozeau:SubsetCoercions} used for Coq's \icode{Program}
tactic. That would allow programmers to specify properties of smart
contract that become obligations in the resulting Coq translation.

\section*{Acknowledgments} We would like to thank the \oaklang/\acorn{} team for
discussions about the language.
\balance
\bibliographystyle{plain}
\bibliography{paper}

\clearpage

\lstset{basicstyle=\footnotesize}
\begin{appendices}
\section{A \icode{Counter} contract in \acorn{}}\label{appendix:counter}
\nobalance
We show how one can obtain the shallow embedding of a
simple contract in \acorn{} and verify some properties using our
framework.\footnote{This example is available in the \texttt{theories/examples/AcornExamples.v} file of our Coq development \url{https://github.com/AU-COBRA/ConCert/tree/artefact}}
First, we start with a \acorn{} program written in concrete syntax.
\begin{lstlisting}
module Counter where
import Prim
import Prod
import Blockchain
import Maybe

data Msg = Inc [Int64] | Dec [Int64]
data CState = CState [Int64, {address}]

definition owner (s :: CState) =
   case s of
     CState _ d -> d

definition balance (s :: CState) =
   case s of
     CState x _ -> x

definition count (c :: Blockchain.ReceiveContext) (s :: CState)
                 (callerOr :: Blockchain.Caller)
                 (amount :: Amount) (msg :: Maybe.Maybe Msg) =
  case msg of
    Maybe.Nothing -> Prod.Pair [CState] [Prim.Transaction] s Prim.TxNone
    Maybe.Just msg0 ->
      case msg0 of
        Inc a ->
          let newS =CState (Prim.plusInt64 (balance s) a) (owner s) in
          Prod.Pair [CState] [Prim.Transaction] newS Prim.TxNone
        Dec a ->
          let newS = CState (Prim.minusInt64 (balance s) a) (owner s)
          in Prod.Pair [CState] [Prim.Transaction] newS Prim.TxNone
\end{lstlisting}
\noindent
From the printing procedure implemented in Haskell that takes the AST of an \acorn{} program as input, we obtain a \acornkernel{} AST of data type definitions and functions corresponding to the \icode{Counter} module.
\begin{lstlisting}
  Definition CounterData := [gdInd "Msg" 0 [("Inc_coq", [(None, tyInd "Z")]); ("Dec_coq", [(None, tyInd "Z")])] false; gdInd "CState" 0 [("CState_coq", [(None, tyInd "Z");(None, tyInd "string")])] false].
\end{lstlisting}
\begin{lstlisting}
  Definition CounterFunctions := [ ("owner", eLambda "x" ((tyInd
    "CState")) (eCase ("CState", []) (tyInd "string") (eRel 0)
    [(pConstr "CState_coq" ["x0";"x1"], eRel 0)])) ;("balance",
    eLambda "x" ((tyInd "CState")) (eCase ("CState", []) (tyInd "Z")
    (eRel 0) [(pConstr "CState_coq" ["x0";"x1"], eRel 1)])); ("count",
    eLambda "x" ((tyInd "ReceiveContext")) (eLambda "x" ((tyInd
    "CState")) (eLambda "x" ((tyInd "Caller")) (eLambda "x" (tyInd
    "nat") (eLambda "x" ((tyApp (tyInd "Maybe") ((tyInd "Msg"))))
    (eCase ("Maybe", [(tyInd "Msg")]) ((tyApp (tyApp (tyInd "Pair")
    ((tyInd "CState"))) ((tyInd "Transaction")))) (eRel 0) [(pConstr
      "Nothing_coq" [], eApp (eApp (eApp (eApp (eConstr "Pair"
      "Pair_coq") (eTy ((tyInd "CState")))) (eTy ((tyInd
      "Transaction")))) (eRel 3)) (eConst "TxNone")); (pConstr
      "Just_coq" ["x0"], eCase ("Msg", []) ((tyApp (tyApp (tyInd
      "Pair") ((tyInd "CState"))) ((tyInd "Transaction")))) (eRel 0)
      [(pConstr "Inc_coq" ["x0"], eLetIn "x" (eApp (eApp (eConstr
        "CState" "CState_coq") (eApp (eApp (eConst "plusInt64") (eApp
        (eConst "balance") (eRel 5))) (eRel 0))) (eApp (eConst
        "owner") (eRel 5))) ((tyInd "CState")) (eApp (eApp (eApp (eApp
        (eConstr "Pair" "Pair_coq") (eTy ((tyInd "CState")))) (eTy
        ((tyInd "Transaction")))) (eRel 0)) (eConst "TxNone")));
        (pConstr "Dec_coq" ["x0"], eLetIn "x" (eApp (eApp (eConstr
        "CState" "CState_coq") (eApp (eApp (eConst "minusInt64") (eApp
        (eConst "balance") (eRel 5))) (eRel 0))) (eApp (eConst
        "owner") (eRel 5))) ((tyInd "CState")) (eApp (eApp (eApp (eApp
        (eConstr "Pair" "Pair_coq") (eTy ((tyInd "CState")))) (eTy
        ((tyInd "Transaction")))) (eRel 0)) (eConst
        "TxNone")))])]))))))].
\end{lstlisting}
\noindent
We assume that primitive data types are available for the translation:
\begin{lstlisting}
  Module Prim.
    Inductive Transaction := txNone.
    Definition plusInt64 := Z.add.
    Definition minusInt64 := Z.sub.
    Definition TxNone := txNone.
  End Prim.
\end{lstlisting}
\noindent
Next, we import \acorn{} ``standard library'' definitions before unquoting. These definitions were also produced by the same printing procedure.
\begin{lstlisting}
  Import Prim.
  Import AcornProd.
  Import AcornMaybe.
  Import AcornBlockchain.
\end{lstlisting}
\noindent
Now, we add all the data types, required for translating
\icode{count}, to the global environment.

\begin{lstlisting}
  Definition gEnv := ((StdLib.Σ ++ AcornMaybe.Data ++ AcornBlockchain.Data ++ AcornProd.Data ++ CounterData)%list).
\end{lstlisting}
We translate data types and the \icode{counter} contract to MetaCoq and unquote them using the template monad.
\begin{lstlisting}
  Run TemplateProgram (translateData CounterData).
  Run TemplateProgram (translateDefs gEnv CounterFunctions).
\end{lstlisting}
\noindent
After that, we can interact with the shallow embedding in a usual way.
\begin{lstlisting}
Print Msg.
(* Inductive Msg : Set :=  Inc_coq : Z -> Msg | Dec_coq : Z -> Msg *)

Print count.
(* count =
   fun _ x0 _ _ x3 =>
   match x3 with
    | @Nothing_coq _ => Pair_coq CState Transaction x0 TxNone
    | @Just_coq _ (Inc_coq x5) =>
       Pair_coq CState Transaction
        (CState_coq (plusInt64 (balance x0) x5) (owner x0)) TxNone
    | @Just_coq _ (Dec_coq x5) =>
       Pair_coq CState Transaction
        (CState_coq (minusInt64 (balance x0) x5) (owner x0)) TxNone
   end
      : ReceiveContext -> CState -> Caller
        -> nat -> Maybe Msg -> Pair CState Transaction *)
\end{lstlisting}
\noindent
Using the shallow embedding, we can show a simple functional correctness property:
\begin{lstlisting}
  Lemma inc_correct init n i fin tx ctx amount caller :
    (* precondition *)
    balance init = n ->
    (* sending "increment" *)
    count ctx init caller amount (Just_coq _ (Inc_coq i)) = Pair_coq _ _ fin tx ->
    (* result *)
    balance fin = n + i.
  Proof. intros Hinit Hrun. inversion Hrun. subst. reflexivity. Qed.
\end{lstlisting}
\end{appendices}

\end{document}